\documentclass[12pt]{amsart}
\usepackage{amsthm}
\usepackage{amsmath}
\usepackage{amstext}
\usepackage{amsfonts}
\usepackage{amssymb}
\usepackage{amsopn}
\usepackage{graphicx}
\usepackage{cite}
\usepackage{bm,bbm}
\usepackage{epsfig,epic}
\usepackage{subfigure}
\usepackage{hyperref}

\usepackage[left=3cm, right=3cm, top=3cm, bottom=4cm]{geometry}

\renewcommand{\phi}{\varphi}
\newcommand{\iy}{\infty}
\renewcommand{\leq}{\leqslant}
\renewcommand{\geq}{\geqslant}

\newcommand{\ket}[1]{| #1 \rangle}
\newcommand{\ketbra}[2]{| #1 \rangle \langle #2 |}

\DeclareMathOperator{\trace}{Tr}
\DeclareMathOperator{\I}{I}
\DeclareMathOperator{\id}{id}

\DeclareMathOperator{\rk}{rk}

\newcommand{\N}{\mathbb{N}}

\newcommand{\C}{\mathbb{C}}
\newcommand{\E}{\mathbb{E}}

\renewcommand{\H}{\mathcal{H}}

\newcommand{\M}{\mathcal{M}}

\renewcommand{\S}{\mathcal{S}}

\newcommand{\cross}{\mathrm{cr}}

\newtheorem{theorem}{Theorem}[section]
\newtheorem{definition}[theorem]{Definition}

\addtocounter{tocdepth}{-1}

\begin{document}

\title[Area law for random graph states]
{\Large  {Area law for random graph states}}

\author{Beno\^it Collins}

\address{B.C.: Department of Mathematics, University of Ottawa, Canada; AIMR, Tohoku University, Sendai; 
CNRS, Lyon 1, France.
{\tt bcollins@uottawa.ca}}

\author{Ion Nechita}
\address{I.N.: CNRS, Laboratoire de Physique Th\'eorique, IRSAMC, 
 Universit\'e de Toulouse, UPS, 31062 Toulouse, France. {\tt nechita@irsamc.ups-tlse.fr}}

\author{Karol \.Zyczkowski}
\address{K.Z.: Institute of Physics, Jagiellonian University,
 Cracow and Center for Theoretical Physics, Polish Academy of Sciences, Warsaw, Poland. 
{\tt karol@tatry.if.uj.edu.pl}}

\begin{abstract}
Random pure states of multi-partite quantum systems, associated with 
arbitrary graphs,  are investigated.
Each vertex of the graph represents a generic interaction between subsystems,
described by a random unitary matrix distributed according to the Haar measure,
while each edge of the graph represents a bi-partite, maximally entangled state.
For any splitting of the graph into two parts we consider the 
corresponding partition of the quantum system and compute the
average entropy of entanglement. First, in the special case where 
the partition does not ``cross'' any vertex of the graph, we show that the area law is satisfied exactly. 
In the general case, we show that the entropy of entanglement obeys an area 
law on average, this time with a correction term that depends on the topologies of the graph and of the partition. The results obtained are applied to the problem of distribution of quantum entanglement in a quantum network with prescribed topology.
\end{abstract}

\maketitle

\tableofcontents

\section{Introduction}

Entanglement in many-body quantum systems is a subject of a considerable recent interest
 \cite{ACL07,WVHC08,ECP10}.
In several physical problems it is important to describe correlations between two 
selected parts of a composed quantum system. For any pure state describing the entire system,
such correlations can be characterized quantitatively
by the {\sl entanglement entropy} $H$, equal to the von Neumann entropy of the 
reduced mixed state obtained by the partial trace over a selected subsystem. 
Note that this quantity depends explicitly on the partition of the 
entire system into subsystems.

Entanglement entropy $H$ does not usually scale proportionally
to the size of the selected region $S$. 
As described in  \cite{AFOV08, ECP10}, for different systems  
the entanglement entropy is approximately proportional to the boundary of this set, 
denoted by $\partial S$. Considering a three dimensional body $S$, 
the size of its boundary is proportional to the {\sl area} of $S$ and not to its volume.
Thus, for any set $S$ of subsystems, the size of its boundary $|\partial S|$ will be called the ``area'' of $S$.
If for a sufficiently large system the leading contribution 
to the entanglement entropy of a given state $|\Psi\rangle$ 
with respect to the partition $\{S,T\}$
is proportional to  the area separating both subsystems,
we will say that the {\sl area law} is satisfied.
In the case of one dimensional systems,
the area law implies that the entropy saturates asymptotically to a constant,
as the area of the boundary consists of two isolated points.

Area laws are studied in context of black hole physics and 
holographic principle \cite{B+86,Sr93}, 
and for ground states of quantum lattice systems with local interactions 
\cite{AEPW02, VLRK03, CC04, Ma09}.
Under some technical assumptions is it possible not only to show
that the area law holds for the ground state of a given two--dimensional model system,  
but also to derive the negative correction term called 
{\sl topological entanglement entropy} \cite{KP06,LW06},
and show that it is universal and depends only on the topology of the interaction.
Investigations of the area law were performed also for subsets with 
fractal boundaries \cite{HLA10}.

For any quantum state of a system of interacting particles, 
described by a lattice or a graph, it is interesting 
to analyze the degree of entanglement between two 
given nodes of the graph \cite{PC+08}. 
This issue is relevant in studies on 
generating entanglement between given nodes of a graph  \cite{PJ+08},
entanglement swapping and quantum repeater systems \cite{ACL07}
and quantum communications in noisy networks \cite{GH+12}.

The aim of this work is to study the area law
for the ensembles of random states associated to a graph. 
The authors have previously introduced \cite{CNZ10}
 an ensemble of pure quantum states for which the structure of interactions
and entanglement are encoded in the graph. In the present work, we study 
the average entropy of entanglement for elements of this ensemble, and we show
that it obeys asymptotically and on average
 area laws of the form
\begin{equation}
	E(\Psi) = H(\rho_S) = |\partial S| \log N  - h_{\Gamma,S} + o(1),
\end{equation}
where $\Psi$ is a graph state, $\rho_S$ is the reduced density operator, 
$|\partial S|$ is the ``area'' of the boundary of the partition $\{S,T\}$, $N$ is 
the dimension of the subsystems and $h_{\Gamma,S}$ is a numerical constant depending
on the graph $\Gamma$ and on the partition $\{S,T\}$. The functionals $E$ and $H$ denote,
respectively, the entropy of entanglement for pure states and the von Neumann
entropy of mixed density matrices.

In previous work, the
 entanglement entropy for a concrete quantum state
of a given composite quantum system  (see e.g. \cite{ECP10}) has been investigated.
Here however, we discuss the statistical properties of 
ensembles of quantum states. We will show 
under which assumptions the area law holds exactly
in our model, while in the opposite case 
the corrections to the area law are derived.

The paper is organized as follows. The model of graph random states from \cite{CNZ10} is
recalled in Section \ref{sec:random-states}, and a detailed example is worked out. 
In Section \ref{sec:adapted}, the notion of adapted marginals of graph states
is introduced and it is shown that for such partitions of the graph
the area law holds exactly.
A more general case of the problem, for which the area law for random graph states
holds only asymptotically is treated in Section \ref{sec:one-vertex}.
In Section \ref{sec:area-general}, a definition of the boundary is proposed in the most general setting.
The general area law is stated in Section  \ref{sec:general}, 
while an application to a transport problem is presented in Section
 \ref{sec:rank-alternative}.
Section \ref{sec:different-dimensions} provides an outline 
of the adjustments to be made to the main results in the case where 
the dimensions of the relevant Hilbert spaces are different.

\section{Random graph states}
\label{sec:random-states}

We recall here the model of {\sl random graph states} introduced in \cite{CNZ10}. 
The most general definition of this model will not be recalled here,
but for completeness we sketch below the main ingredients of the construction. 
For any undirected graph $\Gamma$ consisting of $m$ edges
 $B_1, \ldots, B_m$ and $k$ vertices $V_1, \ldots V_k$,
we associate an ensemble of random pure states $|\Psi\rangle$. 
They describe a quantum system consisting of $n=2m$ particles described in 
the composed Hilbert space $\H_N^{\otimes 2m}$.
The dimension $N$ of any subspace is arbitrary
and in particular we will analyze the asymptotic limit $N\to \infty$.

Each edge of the graph, which connects subsystems labeled by $i$ and $j$,
represents the maximally entangled state, 
$ \ket \Phi^+_{ij} =  \frac{1}{\sqrt{ N}} \sum_{x=1}^{N} |x\rangle_i \otimes |x\rangle _j$,
between two Hilbert spaces $\H_i$ and $\H_j$. 
Each vertex $V$ of the graph, of degree $b$, represents a
generic interaction between $b$ subsystems, described
by a random unitary matrix of size $N^b$, distributed according to the \emph{Haar measure}.
In this way, \emph{independent} random unitary matrices describe unknown interactions in each node, 
which are assumed to be generic. The above discussion can be summarized in the following formula, describing an element from the ensemble:
\begin{equation}\label{eq:def-graph-state}
 |\Psi \rangle =  \left[ \bigotimes_{V \text{ vertex}} 
U_V \right]\left( \bigotimes_{\{i,j\} \text{ edge}} 
|\Phi\rangle^+_{i,j}
 \right),
\end{equation}
where $|\Phi\rangle^+_{i,j}$ are maximally entangled states and $U_V$ are independent Haar unitary operators at each vertex.

Our assumptions differ therefore from the model analyzed in \cite{VC04},
in which edges of the graph denote maximally entangled states of two qubits, while
the vertices represent deterministic local unitary gates or local measurements. 
A more general graph model of quantum networks was investigated in \cite{ACL07,PC+08}
and later studied in context of entanglement percolation \cite{Pe+10b}.
In this version of the model any edge represents a given bipartite state of two qubits, 
while a vertex denotes a deterministic unitary swap gate or a local measurement.
Note that in the system investigated here each dot at the end of an edge
of a graph represents a quantum subsystem described by a $N$ dimensional Hilbert space.
The parameter $N$ is arbitrary, but our results are obtained under
the assumption that the dimension is large.

Let us divide the set of $n=2m$ particles into two disjoint
subsets, labeled by $S$ and $T$. 
The subset $T$ will correspond to subsystems over which
the averaging is performed:
a marginal of the graph state $\ketbra{\Psi}{\Psi}$ defined 
by the partial trace over these subsystems,
\begin{equation}\label{eq:parttrace}
\rho_S = \trace_{T} \ketbra{\Psi}{\Psi},
\end{equation}
forms a mixed state $\rho_S$ supported on the remaining subspaces $S = \{1, \ldots, n\} \setminus T$.
Note that any graph $\Gamma$ and its partition $\mathcal P_\text{trace} = \{S, T\}$
determines an ensemble of random mixed states $\rho_S$.
In \cite{CNZ10} we studied statistical properties of such an ensemble,
which depends on the topology of the graph and on its partition.
Ensemble of random mixed states can be characterized by the
average von Neumann entropy $\E H(\rho_S)=- \E {\rm Tr} \rho_S \log \rho_S$.
By definition this quantity is 
equal to the average entropy of entanglement of the random pure state
$|\Psi\rangle$ with respect to the prescribed 
partition $\mathcal P_\text{trace} = \{S, T\}$.

As a first example for this graphical notation, consider the graph shown in Fig.~\ref{fig:compatible_ex_simple},
which consists of $m=10$ edges (loops and multi-edges are allowed) and $k=5$ vertices. 
Fig.~\ref{fig:non-compatible_ex} shows all $n=2m=20$ subsystems denoted by small black dots
and the random interaction between some of them represented by 
the gray circles at the vertices of the graph. The partition $\{S, T\}$ is indicated graphically by placing 
diagonal crosses on the small black dots corresponding to elements in $T$ (``traced subsystems''), while elements
in $S$ have no crosses on top of them (``surviving subsystems'').

\begin{figure}[htbp]
\centering
\subfigure[]{\label{fig:compatible_ex_simple}\includegraphics{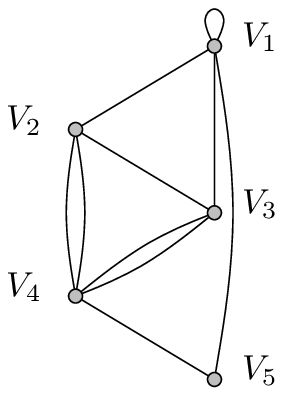}}\qquad\qquad
\subfigure[]{\label{fig:non-compatible_ex}\includegraphics{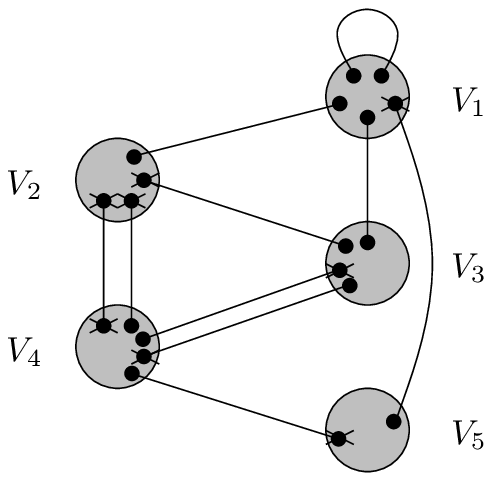}}
\caption{An example of a graph state (a) and one of its marginals (b).}
\label{fig:compatible_example-graph}
\end{figure}

Let us now consider another example, the graph consisting of a single edge, $m=1$
presented in Fig.  \ref{fig:1-edge-graph}. The only edge forms a loop here,
so there is a single vertex $V_1$ only, $k=1$.
The partition splits the $n=2$ partite system into two subsystems $S$ and $T$.
In this case we will relax for a moment the assumption
that the sizes of both subsystems are equal
and will denote them by $N$ and $N'$, respectively. 
In such a case the only parameter of the model is the ratio between
the sizes of both subsystems $c=N/N'$.

\begin{figure}[ht]
\centering
\subfigure[]{\includegraphics{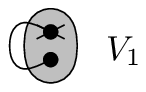}}\quad\quad
\subfigure[]{\includegraphics{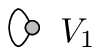}}
\caption{A single loop graph in the standard notation (a) and 
simplified notation (b). This degenerate graph represents bipartite system
($S$ - the black dot and $T$ - crossed dot)
and leads to random mixed states with level density
described by Marchenko--Pastur distributions.}
\label{fig:1-edge-graph}
\end{figure}

As both subsystems are coupled at the vertex $V_1$ by a random unitary matrix 
$U\in U(NN')$
the assumption on the existence of the maximally entangled state becomes
irrelevant here. The partial trace over the subsystem $T$ 
leads to a random mixed state $\rho_S$ of size $N$ \cite{ZS01}.
It can be represented by a normalized Wishart matrix 
$\rho_S=GG^{\dagger}/{\rm Tr}GG^{\dagger}$, 
where $G$ is a rectangular non-hermitian random Ginibre matrix
of size $N \times N'$ (recall that a Ginibre matrix has i.i.d. complex Gaussian entries). The spectral density of the mixed state $\rho_S$
of size $N$ is known to be described asymptotically by the Marchenko--Pastur distribution $\pi_c(x)$, in the following sense (see \cite{nechita} for details): 
\begin{equation}
\text{almost surely, } \quad	\lim_{N \to \infty} \mu_N = \pi_c,
\end{equation}
where the limit corresponds to the weak convergence of probability measures, $\mu_N$ is the empirical distribution of the rescaled eigenvalues of $\rho_S$
\begin{equation}
	\mu_N = \frac{1}{N}\sum_{i=1}^N \delta_{cN\lambda_i(\rho_S)},
\end{equation}
and the parameter $c$ is the asymptotic ratio of dimensions $c = \lim_{N \to \infty} N'/N$. The Marchenko-Pastur distribution is given by
\begin{equation}
\pi_c=\max (1-c,0)\delta_0+\frac{\sqrt{4c-(x-1-c)^2}}{2\pi x} \; \mathbf{1}_{[1+c-2\sqrt{c},1+c+2\sqrt{c}]}(x) \; dx.
\end{equation}
The average von Neumann entropy of a random mixed state $\rho_S$
is given by the integral, 
\begin{equation}
	\E  H(\rho_S)  = \ln N - \int_{}^{} x \log x \; d\pi_c(x) + o(1) = \ln N  - h_c + o(1),
\end{equation}
where $h_c$ is the entropic correction

\begin{equation}
\label{eq:MP-entropy}
h_c  =
\begin{cases}
 \frac{1}{2} + c \log c    \quad & \text{ if } c \geq 1;\\
 \frac{c^2}{2}\quad & \text{ if } 0<c<1.
\end{cases}
\end{equation}

Depending on the desired amount of generality, we are sometimes going to work 
on the model in which all subsystems are described
by Hilbert spaces of the same size $N$.
We will also consider the general version of the model \cite{CNZ10}, in which only the
pairs of subsystems connected by an edge, which describes
a maximally entangled state, have the same dimensions. Hilbert spaces of different dimensions, 
as in the Marchenko--Pastur case treated above, will be allowed. However, we shall ask that these dimensions 
grow at fixed ratios, imposing the asymptotic regime $\mathrm{dim} \H_i = d_i N$, for some 
fixed positive constants $d_i$.

\section{Exact area law for adapted partitions}
\label{sec:adapted}

In this section we show that the area law holds \emph{exactly} for graph states, 
provided that the marginal under consideration satisfies a
 particular condition, called \emph{adaptability}.

Recall that to any graph state we associate two partitions of the set of $n=2m$ subspaces:
 a vertex partition $\mathcal P_\text{vertex}$ 
which encodes the vertices of the graph, and a pair partition $\mathcal P_\text{edge}$
 which encodes the edges (corresponding to maximally entangled states). More precisely, two subsystems $\H_i$ and $\H_j$ belong to the same block of $\mathcal P_\text{vertex}$ if they are attached to the same vertex of the initial graph. Each edge $(i,j)$ of the graph contributes a block of size two $\{i,j\}$ to the edge partition $\mathcal P_\text{edge}$.
 Recall that a marginal~\eqref{eq:parttrace} of a random graph state $\ketbra{\Psi}{\Psi}$
 is specified by a  2-set partition $\mathcal P_\text{trace} = \{S, T\}$. 

\begin{definition}
A marginal $\rho_S$ is called \emph{adapted} if
\begin{equation}
\mathcal P_\text{trace} \geq \mathcal P_\text{vertex}
\end{equation}
for the usual refinement order on partitions. In other words, 
a marginal is adapted if and only if the number of traced out systems in 
each vertex is either zero or maximal. If this is the case, then 
the partition boundary,  which splits the graph into parts $ \{S, T\}$,
does not cross any vertices of the graph.

\end{definition}

Because of the above property, for adapted marginals, we can speak about \emph{traced out vertices}, because if one subsystem of a vertex is traced out, then all
 the other systems of that vertex are also traced out. For the graph state in Figure \ref{fig:compatible_ex_simple} we consider the adapted
 marginal obtained by partial tracing vertices $V_2$ and $V_4$, see Figure \ref{fig:compatible_ex}.

We now define precisely what we mean by area laws \cite{ECP10}. The partition $\{S, T \}$ defines a boundary between 
the set of vertices that are traced out and vertices that survive. In the subsequent sections, the following definition 
will be generalized to take into account non-adapted marginals.

\begin{definition}
The \emph{boundary} of the adapted partition $\{S, T\}$ is defined as the set of all (unoriented) edges $e=\{i_S, j_T\}$ in the graph
 state with the property that $i_S \in S$ and $j_T \in T$. Equivalently, it is the set of edges of the type
 $\includegraphics{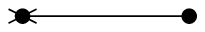}$. The boundary of a partition  shall be denoted by $\partial S$.

The \emph{area} of this boundary is its cardinality $|\partial S |$,  i.e. the number of edges between $S$ and $T$.
\end{definition}

In the example of Figure \ref{fig:compatible_ex}, 
the boundary is represented by a dashed (green) line. The area of the boundary
 in this case is 5: in Figure \ref{fig:compatible_ex}, the boundary line intersects 5 edges.

The main result of this section is that the area law holds \emph{exactly} for adapted marginals of graph states, where we allow arbitrary dimensions of subsystem. Note that, for a given (boundary) edge $\{i,j\}$, we have $d_i = d_j$, the common dimension of the maximally entangled state corresponding to the edge $\{i,j\}$.

\begin{theorem}
Let $\rho_S$ be an \emph{adapted} marginal of a graph state $\ketbra{\Psi}{\Psi}$. Then, the entropy of $\rho_S$ has the following 
\emph{exact, deterministic} value:
\begin{equation}
H(\rho_S) = \log\left(\prod_{\{i,j\} \in \partial S} d_i N\right) = |\partial S| \log N + \log\left(\prod_{\{i,j\} \in \partial S} d_i\right)
\end{equation}
for each value of the size parameter $N$. 
In the particular case where all the Hilbert spaces have dimension $N$ 
(i.e. $d_i=1$ for all $i$), the area law takes the form 
\begin{equation}
H(\rho_S) = |\partial S| \log N.
\end{equation}
\end{theorem}
\begin{proof}
At fixed $N$, the random density matrix $\rho_S$ has the following simple expression 
\begin{equation}
	\rho_S = \left[\bigotimes_{C \in \Pi_\text{vertex}^\text{S}} U_C\right]\left(\left(\bigotimes_{\{i,j\} \in \Pi_\text{edge}^\text{in}}
 \ketbra{\Phi^+_{i,j}}{\Phi^+_{i,j}} \right) \otimes \left(\bigotimes_{\{i,j\} \in \Pi_\text{edge}^\text{out}} \frac{\I_{d_i N}}{d_i N} \right)\right)
\left[\bigotimes_{C \in \Pi_\text{vertex}^\text{S}} U_C\right]^*,
\end{equation}
where $\Pi_\text{vertex}^\text{S}$ is the set of surviving vertices, $\Pi_\text{edge}^\text{in}$ is the set of edges connecting surviving vertices and
 $\Pi_\text{edge}^\text{out} = \partial S$ is the set of edges connecting surviving with traced vertices, i.e. crossed edges.
 Since the entropy is unitarily invariant and additive with respect to tensor products, it is immediate that
\begin{equation}
	H(\rho_S) = \sum_{\{i,j\} \in \partial S} \log(d_i N) = |\partial S| \log N + \log\left(\prod_{\{i,j\} \in \partial S} d_i\right).
\end{equation}
\end{proof}

\begin{figure}[htbp]
\centering
\includegraphics{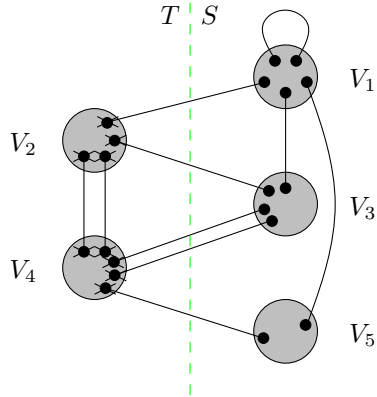}
\caption{An adapted marginal for the graph state in Figure \ref{fig:compatible_ex_simple}. The dashed (green) line represents the \emph{boundary} between the traced--out subsystems $T$ and the surviving subsystems $S$.}
\label{fig:compatible_ex}
\end{figure}
For the system corresponding to the graph shown in Figure \ref{fig:compatible_ex} 
with all subsystems of size $N$  the von Neumann entropy reads
\begin{equation}
H(\rho_S) = 5 \log N .
\end{equation}
This follows from the fact that $\rho_S$ is in this case a unitary conjugation of a maximally mixed state of size $N^5$ with an arbitrary pure state of size $N^6$.

\section{One-vertex marginals}\label{sec:one-vertex}

We shall look now at the simplest situation for which an approximate area law holds.
 We are considering marginals with a unique surviving vertex which may contain traced 
out subsystems. We refer to the results in \cite{CNZ10}, Section 6.3, where the exact same
 situation was studied. In this particular setting, let us introduce some appropriate
 notation for the subsystems of the surviving vertex $V$. The subsystems of $V$ are
 partitioned, on one hand, into surviving $S$ subsystems and traced-out subsystems $T'$
 (note that $T'$ is a subset of the set of all traced systems in the original graph).
 Moreover, the edges of the graph introduce a different partition of $V$, into subsystems
 attached by loops of $V$ (denoted by $F$) and subsystems connected to other vertices, 
which form a set $G$. Hence, three important parameters with respect to the only surviving
 vertex: $|G|$, the number of edges connecting this vertex to other traced out vertices,
 $|T'|$ the number of traced out Hilbert spaces in this vertex and $|S|$, the number of
 surviving Hilbert spaces. For the example presented in Figure
 \ref{fig:1-unitary-marginal-fig}, one has $G=\{1, 6, 7, 8\}$, $T'=\{1, 2, 7, 8\}$ and
 $S=\{3, 4, 5, 6\}$.

\begin{figure}[htbp]
\centering
\includegraphics{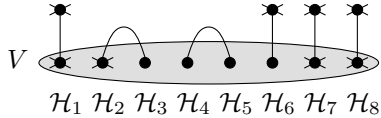}
\caption{Example of a graph state marginal with only one surviving vertex.}
\label{fig:1-unitary-marginal-fig}
\end{figure}

Except for some trivial, degenerate situations ($S=V$, $T=V$ or $G=V$), the entropy of the reduced density matrix is genuinely random, and its asymptotic behavior can be inferred from Theorem 6.4 of \cite{CNZ10}.

\begin{theorem}\label{thm:1-unitary-marginal}
In the limit $N \to \iy$, the average von Neumann entropy of the reduced density matrix $\rho_S$ has the following behavior:
\begin{equation}
	\E H(\rho_S) = o(1) + \begin{cases}
	               	\log(d_S N^{|S|}) & \text{ if } |S| < |T'|+|G|;\\
	\log(d_{T'}d_G N^{|T'|+|G|}) - h_{d_{T'}d_G/d_S} & \text{ if } |S| = |T'|+|G|;\\
	\log(d_{T'}d_G N^{|T'|+|G|})  & \text{ if } |S| > |T'|+|G|.
	               \end{cases}
\end{equation}
where $h_c$ is the entropic correction for a Marchenko-Pastur distribution, defined in equation \eqref{eq:MP-entropy}. In the particular case where $d_i=1$ for all $i$, one has
\begin{equation}\label{eq:entropy-adapted-marginal}
	\E H(\rho_S) = \min\{|S|, |T'| + |G|\}\log N - \frac{1}{2} \delta_{|S|, |T'|+|G|} + o(1).
\end{equation}
\end{theorem}

At this point, it is not obvious how to give an interpretation of Theorem \ref{thm:1-unitary-marginal} and of equation \eqref{eq:entropy-adapted-marginal} in particular, as an area law. Recall that because of the random unitary matrix acting on the vertex $V$, the exact indices of the traced subsystems inside the vertex $V$ are irrelevant (at least in the case of trivial 
relative dimensions). Hence, the topological notion of boundary separating the sets $S$ and $T$ becomes ambiguous. 
In the following section, a more general definition of the boundary of a partition will be given via a combinatorial 
optimization problem.

\section{Defining the boundary surface for a general partition}
\label{sec:area-general}

Before stating and proving the area law in the most general setting, we have to introduce a well-defined 
notion of boundary surface for a general, possibly non-adapted, partition $\{S, T\}$. In this section, we shall restrict our attention to the case where all the subsystems have the same dimension $N$. We shall come back to the general case of arbitrary ratios in Section \ref{sec:different-dimensions}. 

A naive attempt at a definition is to consider the topological notion of boundary that was used in Section \ref{sec:adapted}. In Figure \ref{fig:boundary-ill-defined}, two \emph{a priori} different marginals of the same graph state are represented. Given the fact that independent random unitary matrices act on vertices $V_1$ and $V_2$, one can swap the two Hilbert spaces (say of $V_1$) and leave the distribution of the reduced state invariant. Hence, the two states should have the same statistical properties, although they have different boundary areas: note that the boundary line does not cross any edge on the left hand side picture (null boundary), whereas the right hand side picture has a boundary area of 2 (two intersections). Hence, it is obvious that such a definition of the boundary area is not suitable. The rest of this section is devoted to introducing a new definition for the boundary area of a partition motivated by a strong connection to the maximal flow problem. Note that we restrict ourselves to the case of identical systems,
$d_i=1$, so that $\dim \H_i = N$.

\begin{figure}[htbp]
\centering
\subfigure[]{\label{fig:boundary_1}\includegraphics{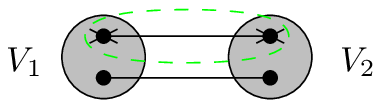}}\qquad\qquad
\subfigure[]{\label{fig:boundary_2}\includegraphics{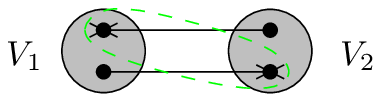}}
\caption{Different assignment of ``crosses'' yield different boundaries. 
The partition on the left has zero boundary (no cuts) and the partition on the right has
 a boundary of $2$: two edges of the graph are cut by the dashed (green) lines.}
\label{fig:boundary-ill-defined}
\end{figure}

Consider, as before, an unoriented graph $\Gamma=(V, E)$ having possibly multiple edges and loops. Consider also a \emph{counting function} $s:V \to \N$,
 such that, for all vertex $v \in V$, $0 \leq s(v) \leq \deg(v)$ and put $t(v) = \deg(v) - s(v)$. These functions will represent the number of surviving (resp. traced) subsystems inside a given vertex. To the pair $(\Gamma, s)$ we shall associate two combinatorial structures: a \emph{network} $\mathcal N_{\Gamma, s}$ and a set of \emph{marked fattened graphs} $\mathcal F_{\Gamma, s}$. In Theorem \ref{thm:graph-theory} we show that the maximal flow in the network $\mathcal N_{\Gamma, s}$ is equal to maximal number of crossings in $\mathcal F_{\Gamma, s}$, providing the connection between statistical properties of the reduced state $\rho_S$ and a combinatorial object related to the partition $\{S, T\}$. This will be the main ingredient in the proof of the general formulation of the area law, Theorem \ref{thm:main}.

Let us first describe the network $\mathcal N_{\Gamma, s}$. Start with the natural network associated
 to $\Gamma$, with the same vertex set as $\Gamma$ and with capacities  $C$ given by the formula
\begin{equation}
C(v ,w)=\begin{cases}
        	\text{number of edges between $v$ and $w$ in $\Gamma$} & \text{ if } v \neq w;\\
		0 & \text{ if } v = w.
        \end{cases}
\end{equation}
To this network, add two distinguished vertices, which were called in \cite{CNZ10} $\id$ and $\gamma$. 
The remaining capacities are defined by $C(\id, v) = t(v) = \deg(v) - s(v)$ and $C(v, \gamma) = s(v)$. 
The network $\mathcal N_{\Gamma, s}$ defined in this way has been shown in \cite[Section 5.2]{CNZ10} to be intimately connected to the statistical properties of random graph states.

We move now to the second combinatorial object associated to the pair $(\Gamma, s)$. First, starting from $\Gamma$, construct the \emph{fattened} graph $\Gamma_\text{fat}$ with vertex set 
\begin{equation}
V_\text{fat} = \bigsqcup_{v \in V} \{v_1, v_2, \ldots, v_{\deg(v)}\}
\end{equation}
and with edge set $E_\text{fat}$ corresponding to the edges of $\Gamma$ in such a way that every two edges are now disjoint. We keep track of the fattening operation by a projection map $f:V_\text{fat} \to V$ defined by $f(v_i) = v$. The fattening operation is depicted in Figure \ref{fig:fat}. A \emph{marking} of fattened graph $\Gamma_\text{fat}$ compatible with the counting function $s$ is a subset $M \subset V_\text{fat}$ such that, for all $v \in V$,
\begin{equation}
	|M \cap f^{-1}(v)| = s(v).
\end{equation}
We write $M \leadsto s$ to denote the fact that the marking $M$ is compatible with the counting function $s$. Note that there always exist compatible markings and that the marking is unique if and only if $s(v)$ is either 0 or maximal for every vertex $v$. This extremal situation corresponds to adapted marginals, which were studied in Section \ref{sec:adapted}. The number of \emph{crossings} of a marking $M$, denoted $\cross(M)$, is the number of edges in $E_\text{fat}$ which connect a marked vertex with an unmarked one
\begin{equation}
	\cross(\Gamma_\text{fat}, M) = |\{i,j\} \in E_\text{fat} \; | \; (i \in M , j \notin M) \text{ or } (i \notin M , j \in M)\}|.
\end{equation}
For the first marking in Fig.  \ref{fig:fat}b, 
one has $\cross(\Gamma_\text{fat}, M_1) = 6$, while for the one shown in Fig.  \ref{fig:fat}c, the number of crossings reads
 $\cross(\Gamma_\text{fat}, M_2) = 8$.

\begin{figure}[ht]
\centering
\subfigure[]{\includegraphics{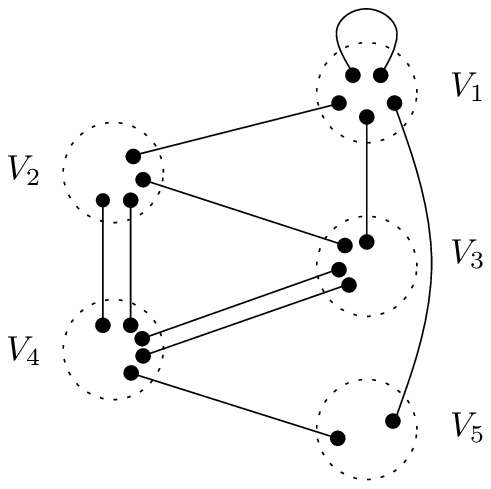}}\\
\subfigure[]{\includegraphics{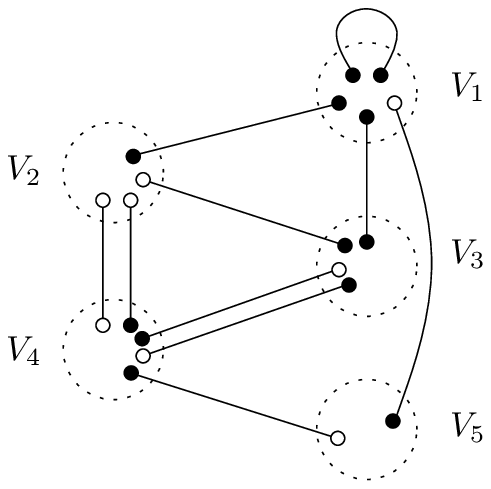}} \qquad\qquad
\subfigure[]{\includegraphics{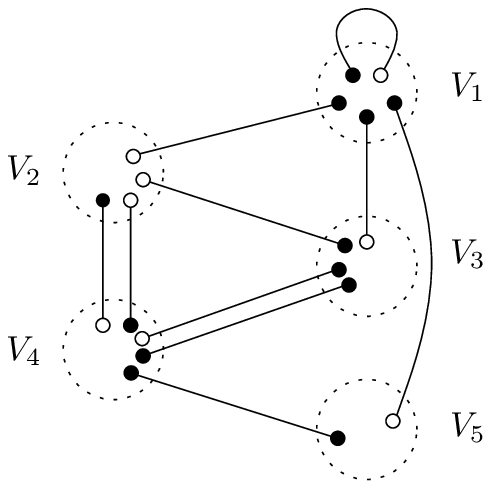}}
\caption{The fattening of the graph in Figure \ref{fig:compatible_example-graph} (a)
 and two markings,  $M_{1}$ in panel (b) and $M_2$ in (c),
 compatible with the count function (number of marked
subsystems in each vertex),  $s(V_1)=s(V_3)=s(V_5)=1$, $s(V_2)=3$, and
 $s(V_4)=2$. Marked subsystems are represented by empty dots,
  and the number of crossings  $\cross$ is equal to the number of edges 
 with one vertex filled and the other one empty.}
\label{fig:fat}
\end{figure}

The main result of this section is the graph-theoretical Theorem \ref{thm:graph-theory}, which makes use of the following important definition.

\begin{definition}\label{def:area}
For a graph $\Gamma$ and a partition $\{S,T\}$ of its associated nodes, define the \emph{area of the boundary} (or, simply, \emph{area}) of the partition $\{S,T\}$ as
\begin{equation}\label{eq:area}
|\partial S| = \max_{M \leadsto s} \cross(\Gamma_\text{fat}, M),
\end{equation}
where the maximum is taken over all the markings $M$ of the set of vertices of the fattened graph $\Gamma_\text{fat}$ compatible with the counting function $s$.
\end{definition}

\begin{theorem}\label{thm:graph-theory}
For any graph $\Gamma$, the maximal flow $X_{\Gamma, S}$ in the network $\mathcal N_{\Gamma,S}$ is equal to the area of the boundary of the partition $\{S,T\}$
\begin{equation}
X_{\Gamma, S} = |\partial S|.
\end{equation}
\end{theorem}
\begin{proof}
We shall prove inequalities in both directions. First, consider a compatible marking $M \leadsto s$. We shall construct a set of augmenting paths in the network having a total flow of  $\cross(\Gamma_\text{fat}, M)$. For every crossing edge attached to a single vertex $v$ of $\Gamma$, consider the augmenting path $\id \to v \to \gamma$. For a crossing edge $e = (v, w)$, where the black dot is in vertex $v$ and the empty dot is in vertex $w$, consider the augmenting path $\id \to v \to w \to \gamma$. In this way, to each crossing edge, we associate a unit of flow from $\id$ to $\gamma$, proving thus $\cross(\Gamma_\text{fat}, M) \leq X_{\Gamma,s}$. Maximizing over all compatible markings $M$ proves the first inequality. 

Let us move now to proving the other direction. To this end, consider a set of augmenting paths in the network $\mathcal N_{\Gamma, s}$, achieving the maximal flow $X_{\Gamma, s}$. Let 
$$ \id \to v_1 \to v_2 \to \cdots \to v_k \to \gamma$$
be such an augmenting path of length $k \geq 1$.  If $k=1$, choose an edge $(v_1,v_1)$ in $\Gamma_\text{fat}$ and mark one of its vertices as filled and the other one as empty. Otherwise, one can find the edges $(v_1,v_2)$, $(v_2,v_3)$, $\ldots$, $(v_{k-1},v_k)$ in the fattened graph $\Gamma_\text{fat}$. Color these edges in the following way:
\begin{itemize}
\item $(v_1,v_2)$ : $v_1$ filled, $v_2$ empty; 
\item $(v_2,v_3)$ : $v_2$ empty, $v_3$ empty;  
\item $\cdots$ 
\item $(v_{k-1},v_k)$ : $v_{k-1}$ empty, $v_k$ empty.
\end{itemize}
In this way, for each augmenting path of unit flow (one can always assume this, at the cost of repeating edges), one assigns a unique crossing in the fattened graph. It follows that, for this marking $M$, one has $\cross(\Gamma_\text{fat}, M) = X_{\Gamma,s}$. 
This proves the theorem.
\end{proof}

\section{A general area law for graph states}\label{sec:general}

This section contains the proof of the main result of the paper, Theorem \ref{thm:main}. 

\begin{theorem}[Area law for random graph states]\label{thm:main}
Let $\rho_S$ be the marginal $\{S, T\}$ of a graph state $\Gamma$. Then, as $N \to \iy$, the  \emph{area law} holds, in the following sense
\begin{equation}\label{eq:area-law}
	\E H(\rho_S) = |\partial S| \log N - h_{\Gamma, S} + o(1),
\end{equation}
where $|\partial S|$ is the area of the boundary of the partition $\{S,T\}$ defined in \ref{def:area} and $h_{\Gamma, S}$ is a positive constant, depending on the topology of the network $\mathcal N_{\Gamma,S}$ (and independent of $N$).
\end{theorem}
\begin{proof}
The idea is to combine moment computations from \cite{CNZ10} for the random matrix $\rho_S$ with the combinatorial identity proved in Theorem \ref{thm:graph-theory}. Recall the following moment formula from \cite[Theorem 5.5]{CNZ10}:
\begin{equation}
	\forall p \geq 1, \quad \E \trace(\rho_S^p) = N^{-X_{\Gamma, S}(p - 1)}( |B_p| + o(1)),
\end{equation}
where $B_p \subset NC(p)^k$ is the subset of non-crossing partitions (or, equivalently, geodesic permutations) corresponding to the augmenting paths leading to the maximal flow $X_{\Gamma, S}$ in the network $\mathcal N_{\Gamma, S}$ (for details, see \cite{CNZ10}). We can restate the above asymptotic expression as a limit: (we write simply $X = {X_{\Gamma, S}}$)
\begin{equation}
\forall p \geq 1, \quad  \lim_{N \to \infty} \E \frac{1}{N^X}  \trace\left[ ( N^X \rho_S)^p \right] = |B_p|.
\end{equation}
In other words, the measures 
\begin{equation}
\mu_N = \E \frac{1}{N^X}  \sum_{i=1}^{N^X} \delta_{ N^X \lambda_i(\rho_S)}
\end{equation}
have limiting moments given by $ |B_p|$. Note that there is a unique probability measure having moments $|B_p|$, since one has the bound $|B_p| \leq \mathrm{Cat}_p^k \leq (4^k)^p$ which is exponential, and thus, by Carleman's condition, these moments uniquely define the probability measure. Since the limit points of the tight sequence of measures $\mu_N$ are uniquely determined by their moments, it follows that \cite[Theorem C.9]{agz} $\mu_N$ converges weakly towards a measure $\mu$ satisfying
$$\forall p \geq 1, \qquad \int x^p  d\mu(x) = |B_p|.$$
Finally, one has 
\begin{equation}
	\E H(\rho_S) =X \log N - h_{\Gamma, S} + o(1),
\end{equation}
with a correction term equal to the entropy of the
asymptotic measure $\mu$,
\begin{equation}
 h_{\Gamma, S} = \int x \log x d\mu(x).
 \end{equation}
\end{proof}
If the measure $\mu_N$ converges to the Marchenko-Pastur distribution 
with parameter $c$, the above term coincides with the
entropy $h_c$ defined in (\ref{eq:MP-entropy}).

\section{Rank of random graph states and a transport problem}\label{sec:rank-alternative}

We start by looking at a linear algebra problem, the maximum rank of a marginal of a graph state, over the set of unitary operations on vertices. 

\begin{theorem}\label{thm:rank}
The maximum rank of a graph state marginal $\rho_S$ is the area of the boundary of the partition $\{S,T\}$
\begin{equation}\label{eq:rank}
\max_{U_1 , \ldots, U_k} \mathrm{rk} \rho_{S} = |\partial S|.
\end{equation}
Moreover, this maximum can be achieved by choosing the $U_i$ to be permutation matrices and the state $\rho_S^*$ which achieves the maximum can be taken maximally mixed. 
\end{theorem}

The above result shows that the \emph{maximum} achievable rank of a graph state is equal, asymptotically, to the rank of a \emph{random} marginal. We shall investigate further this property by presenting an alternative formulation to our problem, bridging together the max-flow, max-crossings and the rank aspects. It also has the advantage of being ``operational''. 
The notation below mirrors the one in the rest of the paper and in \cite{CNZ10}. 
Note that there is no randomness in the problem below, and $N$ can be arbitrary.
We consider below the simplest one--qubit case, $N=2$. 

Consider the following problem (see Figure \ref{fig:V3} for an example). 

\begin{figure}[htbp]
\centering
\subfigure[]{\includegraphics{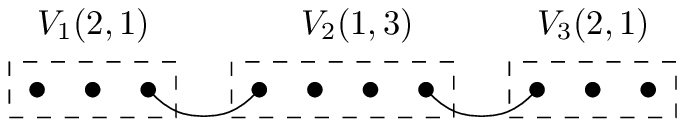}}\\
\subfigure[]{\includegraphics{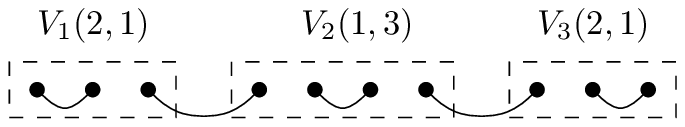}}\\
\subfigure[]{\includegraphics{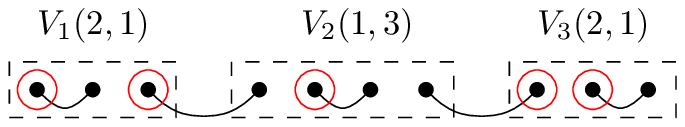}}\\
\caption{An instance of the problem, where $S_i$ and $T_i$ are noted in parenthesis after each $V_i$.
 In (b), we create locally singlet states. The states sent to $A$ are circled in red in (c), the other 
ones being sent to $B$. Note that at each site we consider $S_i+T_i$ particles.}
\label{fig:V3}
\end{figure}

\noindent \textbf{Problem.} The company \emph{EntanglementFactory} has $k$ research facilities around the world, let us call them $V_1, \ldots, V_k$. As a result of past experiments, entangled states are shared between pairs of these laboratories, as follows : labs $V_i$ and $V_j$ share $E_{ij}$ singlet states $| \Phi^+ \rangle \in \mathbb C^N \otimes \mathbb C^N$. Each facility $V_i$ has an unlimited supply of extra $N$-dits which are not entangled with anything else. One day, company \emph{EntanglementFactory} receives an order from the company \emph{WantEntanglement} for a supply of entangled states. Company \emph{WantEntanglement} has two research facilities $A$ and $B$ and offers to ship $N$-dits from each factory $V_i$ to $A$ and/or $B$, as follows: $S_i$ $N$-dits can be shipped from $V_i$ to $A$ and $T_i$ $N$-dits can be shipped from $V_i$ to $B$. Company \emph{WantEntanglement} pays \$10
 for each unit of entanglement it will have between stations $A$ and $B$. What is the maximal profit the company \emph{EntanglementFactory} can make, just by using local unitary matrices at each site $V_i$ ?

We shall answer this question in the following three scenarios, each situation imposing some physical restrictions or liberties on the system.

\noindent \textbf{Scenario 1 : No initial entanglement.} Suppose that all the entangled particles shared between pairs of $V_i$'s are lost. The best we can do is to create locally, at each $V_i$, maximally entangled (singlet) states and to ship one half to $A$ and the other half to $B$. The number of entangled pairs between $A$ and $B$ will then be
$$Y_1 = \sum_i \min(S_i, T_i).$$

\noindent \textbf{Scenario 2 : Global operations are allowed.} Suppose that company \emph{EntanglementFactory} has the ability of performing \emph{global} operations on all of its facilities $V_i$. Then it is easy to produce a maximally entangled state between $A$ and $B$:
$$Y_2 = \min(\sum_i S_i, \sum_i T_i).$$

\noindent \textbf{Scenario 3 : Entanglement with local operations.} Without any assumptions, we show that
$$Y_3 = X_{\Gamma, S}.$$

First, notice that the problem in this case is just as a restatement of the rank theorem discussed earlier,
 before taking the partial trace of the state. The facilities $A$ and $B$ define the partition with respect to which 
the partial trace is considered. The entropy of entanglement of the pure state shared between $A$ and $B$ is just the
 von Neumann entropy of the reduced density matrix $\rho_S$. 
 Note that the following inequality concerning the number of entangled pairs 
in different scenarios holds $ Y_1 \leq Y_3 \leq Y_2.$

One can characterize entanglement by using the generalized R\'enyi entropy $H_q(\rho):= \frac{1}{1-q} \ln {\rm Tr} \rho^q$.
Note that in the limit the R\'enyi parameter $q$ tends to unity
this expression reduces to the von Neumann entropy, $\lim_{q\to 1} H_q(\rho)=H(\rho)$.
Furthermore, the rank $\rk$  of a matrix is given by the generalized
entropy of order zero,  $\log \rk \rho = H_0(\rho)$.

To summarize, the results on the rank and on the R\'enyi entropy $H_q$  
have the following translation 

\begin{theorem}
There exist local unitary operations $U_i$ (which can be taken to be permutation matrices, 
i.e. at each site $V_i$ we just have to say where each particle goes, to $A$ or $B$) such that
$$H_0(\rho) = H_1(\rho) = H_p(\rho) = X_{\Gamma, S} \log N ,$$
for any $q\ge 0$. In other words, $\rho$ is essentially a maximally mixed state. 
Moreover, the permutation matrices involved can be computed efficiently, using a flow algorithm.
\end{theorem}

In the random case (suppose the engineers of company \emph{EntanglementFactory} are on vacation, 
so the staff decides to implement at each site $V_i$ random, independent local unitary transformations), we have the following result, which is a restatement of Theorem \ref{thm:main}.

\begin{theorem}
For independent random Haar unitary matrices $U_i$, we have, almost surely as $N \to \infty$,
$$H_0(\rho) = X_{\Gamma, S} \log N$$
and 
$$\E H_1(\rho) = X_{\Gamma, S} \log N - h_{\Gamma, S} + o(1).$$
\end{theorem}

Comparing the two results above, one concludes that the random choice is nearly optimal.

\section{Some results for different subsystem dimensions}
\label{sec:different-dimensions}
 In this final section, we analyze simple graphs in the general setting, where we allow subsystems to have different dimensions. Although we can not state a general area law for marginals of such graph states, we perform direct computations in some simple settings using the full machinery developed in \cite{CNZ10}. Our main tool is the following result, valid for random graph states with subsystems of dimensions $d_iN$.

\begin{theorem}[see {\cite[Theorem 5.4]{CNZ10}}]\label{thm:moments-network}
The asymptotic moments of a graph state marginal $\rho_S$ are given by the formula
\begin{align}
\E \trace(\rho_S^p) &=(1+o(1)) N^{-X(p-1)}\!\!\!\!\!\!\!\! \sum_{(\beta_1, \ldots, \beta_k) \in B} \prod_{i=1}^k \left( d_{S_i} \right)^{\#(\gamma^{-1} \beta_i)} \prod_{i=1}^k \left( d_{T_i} \right)^{\#\beta_i} \\
\notag & \qquad \cdot
 \prod_{1 \leq i < j \leq k} \left( d_{E_{ij}} \right)^{\#(\beta_i^{-1} \beta_j)-p} \prod_{i=1}^k d_{C_i}^{-p},
\end{align}
where $B$ is a set of permutations that can be computed from the network $\mathcal N_{\Gamma, S}$ associated to the marginal and $X$ is the maximal flow $\id \to \gamma$ in the same network. For each block $C_i$ of $\Pi_\text{vertex}$, we write $S_i = S \cap C_i$ and $T_i = T \cap C_i$. Finally, $E_{ij}$ is the set of edges going from vertex $i$ to vertex $j$. 
\end{theorem}

The above general expression for the average $p$-th moment allows us to analyze ensembles
of quantum states corresponding to the following exemplary graphs, each containing $4$ subsystems, but having different geometry.

\subsection*{Black holes graph}

We are going to discuss here a simple graph described by two edges
joined in one vertex, used to model the trans--horizon entanglement
during the process of evaporation of a black hole \cite{BSZ09,BP11}. 
The entire system is thus composed out of four subsystems,
two of which have the same dimension equal to $d_1N$,
while the other two have the dimensions $d_2N$. The ratios
$d_{1,2}$ are treated as parameters of the model. We shall consider two marginals of 
this graph state, see Figure \ref{fig:black-hole-graph}. In both cases, the network associated to the marginal 
is the same, and has a maximum flow $X=2$. Moreover, the set $B$ of permutations achieving this maximum flow is 
\begin{equation}
	B = \{(\beta_1,\beta_2,\beta_3) \in \S_p^3 \, : \, \id = \beta_1 \leq \beta_2 \leq \beta_3 = \gamma\}.
\end{equation}

\begin{figure}[htbp]
\centering
\subfigure[]{\includegraphics{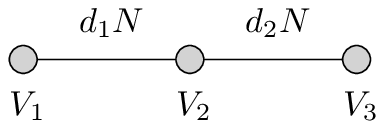}}\qquad
\subfigure[]{\includegraphics{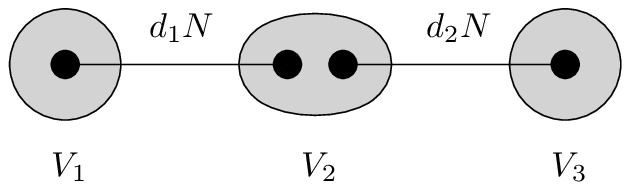}}\\
\subfigure[]{\includegraphics{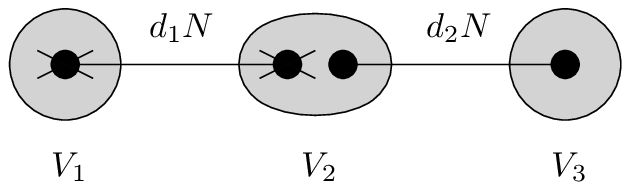}}\qquad
\subfigure[]{\includegraphics{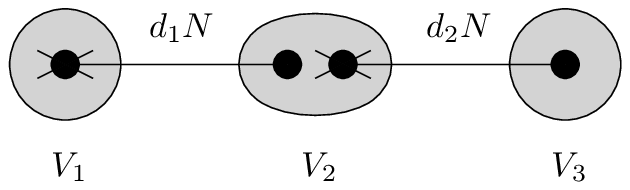}}\\
\subfigure[]{\includegraphics{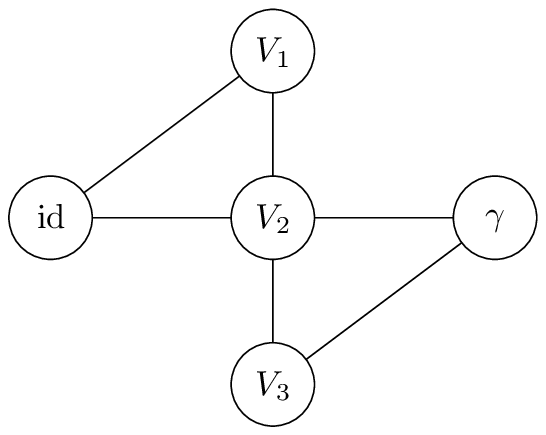}}
\caption{The black hole graph, in the simplified notation (a) and the diagram with 4 subsystems (b). In the middle row, two of its marginals: on the left (c), $d_1N$-sized subsystems are traced out and on the right (d), subsystem of both dimensions are traced out. In the bottom row, the network associated to both marginals (e).}
\label{fig:black-hole-graph}
\end{figure}

\subsubsection*{First case : subsystems of size $d_1N$ are traced}
Applying Theorem \ref{thm:moments-network} to this setting, we obtain 
\begin{align}
	\E \trace\rho_S^p &=(1+o(1)) N^{-2(p-1)} \sum_{\id = \beta_1 \leq \beta_2 \leq \beta_3 = \gamma} d_2^{\#(\gamma^{-1}\beta_2)} d_2^{\#(\gamma^{-1}\beta_3)} d_1^{\#(\beta_1)} d_1^{\#(\beta_2)} \\
\notag	& \qquad \cdot d_1^{\#(\beta_1^{-1}\beta_2) - p} d_2^{\#(\beta_2^{-1}\beta_3) - p} (d_1^2 d_2^2)^{-p} \\
\notag &= (1+o(1))  N^{-2(p-1)} d_1^{-2p} d_2^2 \sum_{\sigma \in NC(p)} \left( \frac{d_1}{d_2}\right)^{2\#\sigma}.
\end{align}
In other words 
\begin{equation}
	\lim_{N \to \infty} \frac{1}{d_2^2N^2}\E \trace(d_1^2N^2\rho_S)^p = \sum_{\sigma \in NC(p)} \left( \frac{d_1}{d_2}\right)^{2\#\sigma}.
\end{equation}
Thus, the rescaled random matrix $d_1^2N^2\rho_S \in \M_{d_2^2N^2}(\C)$ converges in moments to the Marchenko--Pastur distribution
 of parameter $d_1^2 / d_2^2$. Let us now compute the average entropy of this random matrix. Use
\begin{equation}
	\lim_{N \to \infty} \frac{1}{d_2^2N^2} \E H(d_1^2N^2\rho_S) = -h_{d_1^2 / d_2^2}
\end{equation}
to show that
\begin{align}
\label{eq:entropy-black-hole-1}	\E H(\rho_S) &= o(1) + \begin{cases}
	                   	\log(d_2^2N^2) - \frac{d_2^2}{2d_1^2} & \text{ if } d_1 \geq d_2\\
				\log(d_1^2N^2) - \frac{d_1^2}{2d_2^2} & \text{ if } d_1 < d_2,
	                   \end{cases}\\
\notag	                   &=\log(d^2N^2) - \frac{d^2}{2D^2} +o(1),
\end{align}
where $d=\min(d_1, d_2)$ and $D = \max(d_1, d_2)$. Note that the formula above is symmetric in $d_1$ and $d_2$, a consequence of the fact that the non-zero spectra of the two reduced density operators of a pure state are identical.

\subsubsection*{Second case : subsystems of both sizes are traced}
In this case, the moment formula reads
\begin{equation}
	\E \trace\rho_S^p  = (1+o(1)) N^{-2(p-1)} (d_1d_2)^{-p+1} \sum_{\id = \beta_1 \leq \beta_2 \leq \beta_3 = \gamma} 1.
\end{equation}
Thus, the rescaled random matrix $d_1d_2N^2\rho \in \M_{d_1d_2N^2}(\C)$ converges in moments to the Marchenko--Pastur
 distribution of parameter $1$. The entropy computation in this case is easier:
\begin{equation}
	\E H(\rho_S) = 	\log(d_1d_2N^2) - \frac{1}{2} + o(1).
\end{equation}

Note that the two entropy formulas agree in the case $d_1 = d_2$. 

\subsection*{Double line graph - the oxygen molecule $O_2$}

The graph presented in Figure \ref{fig:oxygen-graph}, which can be symbolically represented by  $O=O$, might be
interpreted as the oxygen molecule.
 We will discuss here the general version of the model 
in which there are two pairs of subsystems of size $d_1N$ and $d_2N$ respectively, and look at two different marginals. The two marginals correspond to the same network, which has a maximal flow $X=2$. The set $B$ of permutations achieving this maximum flow is 
\begin{align}
	B &= \{(\beta_1,\beta_2) \in \S_p^2 \, : \, \id \leq \beta_1 \leq \beta_2 \leq  \gamma\ \text{ and } \id \leq \beta_2 \leq \beta_1 \leq  \gamma\}\\
	\notag &=  \{(\beta_1,\beta_2) \in \S_p^2 \, : \, \id \leq \beta_1 = \beta_2 \leq  \gamma\ \},
\end{align}
which is in bijection with the set found for the ``black-hole'' graph.

\begin{figure}[htbp]
\centering
\subfigure[]{\includegraphics{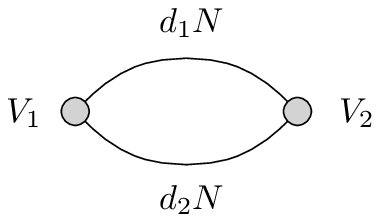}}\qquad
\subfigure[]{\includegraphics{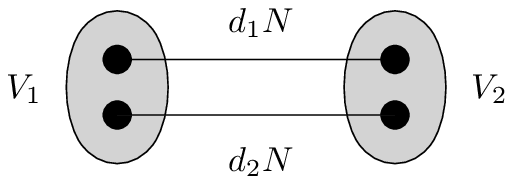}}\\
\subfigure[]{\includegraphics{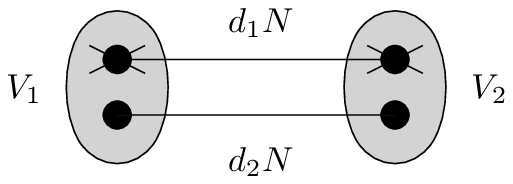}}\qquad
\subfigure[]{\includegraphics{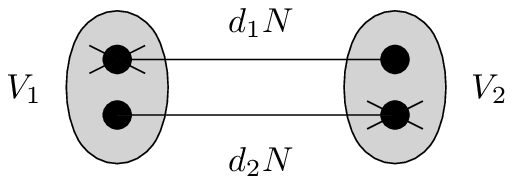}}\\
\subfigure[]{\includegraphics{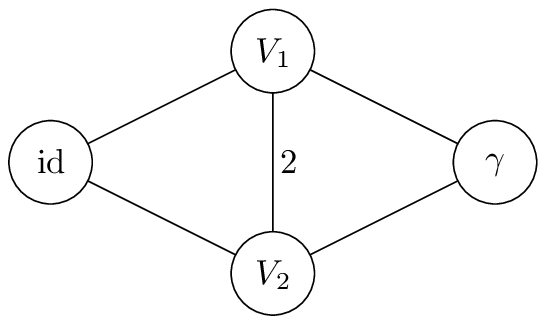}}
\caption{The oxygen graph, in the simplified notation (a) and the diagram with 4 subsystems (b). In the middle row, two of its marginals: on the left (c), $d_1N$-sized subsystems are traced out and on the right (d), subsystem of both dimensions are traced out. In the bottom row, the network associated to both marginals (e).}
\label{fig:oxygen-graph}
\end{figure}

\subsubsection*{First case : subsystems of size $d_1N$ are traced} 
Proceeding as in the case of the ``black-hole'' graph, we obtain the same result as in the first case above. Thus, as before
\begin{equation}
	\E H(\rho_S) = \log(d^2N^2) - \frac{d^2}{2D^2} +o(1),
\end{equation}
where $d=\min(d_1, d_2)$ and $D = \max(d_1, d_2)$.

\subsubsection*{Second case : subsystems of both sizes are traced}
Again, we obtain the same result as in the corresponding case of the ``black-hole'' graph
\begin{equation}
	\E H(\rho_S) = \log(d_1d_2N^2) - \frac{1}{2} +o(1).
\end{equation}

To conclude, note that in general, the relative size of the traced out systems matters (unless $d_1=d_2$). At this point, we have no interpretation for the surprising fact that the two graphs studied here yield the same output entropies.

\section{Perspectives and open questions}

In this work, we study the structured model of random pure quantum states introduced in \cite{CNZ10} from the perspective of area laws. We showed, in the situation where the vertex size is constant, that the entropy of entanglement satisfies, on average, an area law, for a suitable definition of surface area. Indeed, since we are dealing with unitary mixing at each vertex, the usual notion of area does not make sense, so one defines surface area via a combinatorial optimization procedure. In the final section of the paper, we studied some situations where Hilbert space dimension varies, in the case of very simple graphs. Unfortunately, our current methods (area defined via combinatorial optimization) are not adapted anymore, and some further work is necessary to establish an area law in this more general setting. 

A mathematical improvement over the current results would be to obtain estimates for the probabilities of failure of the announced area laws. Indeed, our results focus on average quantities and it would be interesting for to derive \emph{large deviations bounds} for the entropy of entanglement at fixed (but large) Hilbert space dimension $N$. 

Another direction for future work would be to continue the project started in \cite{CNZ10} and to analyze different models of structured entanglement, motivated by solid state physics. Indeed, our starting assumption is that the initial entanglement between vertices is encoded by maximally entangled states. It would be natural to drop this assumption and to work with \emph{generic entanglement}, that could be generated, say, by associating to graph edges an independent set of unitary matrices. This would render the graph state model more symmetric and make it more realistic. 

\bigskip 

\noindent\textbf{Acknowledgments.}
It is a pleasure to thank Pawe{\l} Kondratiuk 
for fruitful discussions on random graph states, and the anonymous referees for several useful remarks and comments.
Financial support by the Polish National Centre of Science
under the grant number DEC-2011/02/A/ST1/00119 
and by the Deutsche Forschungsgemeinschaft under the project
SFB Transregio--12  is  gratefully acknowledged. 
I.~N.~ acknowledges financial support from the ANR project \mbox{OSvsQPI} and the PEPS-ICQ CNRS project \mbox{Cogit}.
The research of B.~C.~ was partly supported by an NSERC discovery grant, an Ontario's ERA and AIMR.

\end{document}